\documentclass{LMCS}
\usepackage[all]{xy}
\usepackage{amssymb}
\usepackage{amsmath}
\usepackage{wrapfig}
\usepackage{graphics,enumerate,hyperref}

\newcommand{\pfrule}[2]{{{\frac{\hbox{{$#1$}}}
 {{\strut \hbox{{$#2$}}}}}}}
\newcommand{\onthetop}[2]{{\renewcommand{\arraystretch}{1.1}
\begin{tabular}{c}\mbox{\!#1\!}\\\mbox{\!#2\!}\end{tabular}}}
\newcommand{\Empty}{\ensuremath{\mbox{\sc Empty}}}
\newcommand{\tr}{\ensuremath{tt}}
\newcommand{\Qe}{\ensuremath{Q_{\exists}}}
\newcommand{\Qu}{\ensuremath{Q_{\forall}}}
\newcommand{\act}{\mbox{${A}${\sl ct}}}
\newcommand{\actc}{\ensuremath{\act_c}}
\newcommand{\actr}{\ensuremath{\act_r}}
\newcommand{\acti}{\ensuremath{\act_i}}
\newcommand{\simul}{\ensuremath{\sqsubseteq_s}}
\newcommand{\simuleq}{\ensuremath{=_s}}
\newcommand{\csimul}{\ensuremath{\sqsubseteq_{cs}}}
\newcommand{\csimuleq}{\ensuremath{=_{cs}}}
\newcommand{\rsimul}{\ensuremath{\sqsubseteq_{rs}}}
\newcommand{\rsimuleq}{\ensuremath{=_{rs}}}
\newcommand{\tnsimul}{\ensuremath{\sqsubseteq_{2s}}}
\newcommand{\tnsimuleq}{\ensuremath{=_{2s}}}

\newcommand{\pre}{\mathit{pre}}
\newcommand{\post}{\mathit{post}}
\newcommand{\defin}{\stackrel{\rm def}{=}}

\newcommand{\goes}[1]{\ensuremath{\stackrel{#1}{\longrightarrow}}}
\newcommand{\goesw}[1]{\ensuremath{\stackrel{#1}{\Longrightarrow}}}

\newcommand{\must}[1]{[ #1 ]}
\newcommand{\may}[1]{\langle #1 \rangle}

\def\doi{5 (1:2) 2009}
\lmcsheading%
{\doi}
{1--22}
{}
{}
{Apr.~17, 2008}
{Jan.~26, 2009}
{}   

\begin{document}

\title[Beyond Language Equivalence on vPDA]
{Beyond Language Equivalence on \\ Visibly Pushdown Automata\rsuper*}
\author[J.~Srba]{Ji\v{r}\'{\i} Srba}
\address{                                                           
Department of Computer Science, 
Aalborg University, Selma Lagerl\"{o}fs Vej 300, 9220 Aalborg East, Denmark} 
\email{srba@cs.aau.dk}
\thanks{The author is supported
in part by Institute for Theoretical Computer Science, project No.~1M0545.}

\keywords{visibly pushdown automata, bisimilarity checking, regularity,
 mu-calculus}
\subjclass{F.4.3}
\titlecomment{{\lsuper*}This is a full and extended version of a paper published
at CSL'06.}

\begin{abstract}
We study (bi)simulation-like
preorder/equivalence checking on visibly pushdown automata,
visibly BPA (Basic Process Algebra) and
visibly one-counter automata. We describe generic methods 
for proving complexity upper and lower bounds for a number
of studied preorders and equivalences like 
simulation, completed simulation, ready simulation,
2-nested simulation preorders/equivalences and bisimulation
equivalence. 
Our main results are that all the mentioned equivalences
and preorders are EXPTIME-complete on visibly pushdown automata,
PSPACE-complete on visibly one-counter automata and P-complete on
visibly BPA. 
Our PSPACE lower bound for visibly one-counter automata
improves also the previously known DP-hardness results
for ordinary one-counter automata and one-counter nets.
Finally, we study regularity checking problems 
for visibly pushdown automata and show that they can be decided
in polynomial time.
\end{abstract}

\maketitle

\section{Introduction}
Visibly pushdown languages were introduced by Alur and Madhusudan
in~\cite{AM:vPDA:04} as a subclass of context-free languages
suitable for formal program analysis, yet tractable and with nice 
closure properties like the class of regular languages. 
Visibly pushdown languages are recognized by visibly pushdown automata
whose stack behaviour is determined by the input symbol. 
If the symbol belongs to the category of \emph{call actions} then
the automaton must push, if it belongs to \emph{return actions} then
the automaton must pop, otherwise (for the \emph{internal actions})
it may not change the stack height. In~\cite{AM:vPDA:04}
it is shown that the class of visibly pushdown languages is closed
under intersection, union, complementation,
renaming, concatenation and Kleene star. A number of decision problems like
universality, language equivalence and language inclusion,
which are undecidable for context-free languages, become  
EXPTIME-complete for visibly pushdown languages.

Recently, visibly pushdown languages have been
intensively studied and applied to e.g.
program analysis~\cite{AEM:04}, XML processing~\cite{Pitcher:05}
and the language theory of this class has been further 
investigated in~\cite{AKMV:06,BLS:06}.
Researches also studied visibly pushdown games~\cite{LMS:FSTTCS:04}. 
Some recent results show, for example, the application of a variant of 
visibly pushdown automata  
for proving decidability of contextual equivalence (and other problems)
for the third-order fragment of Idealized Algol~\cite{MW:05}.
Several strict extensions of visibly pushdown automata, which still preserve
some of their pleasing language properties, have been introduced 
in~\cite{Caucal:DLT:06, FisPnu:01, NS:MFCS:07}. Note that the extension 
introduced in~\cite{FisPnu:01} does not use the terminology of
visibly pushdown automata, while it still employs similar ideas.

In this article we study visibly pushdown automata from
a different perspective. Rather than as language acceptors, we consider
visibly pushdown automata as devices 
that generate infinite state labelled graphs
and we study the questions of decidability of behavioral equivalences
and preorders on this class. Such questions were previously intensively
studied on different classes of infinite state systems, motivated by the
\emph{equivalence checking}
approach where a given implementation and specification
of a system are compared with respect to a suitable notion of behavioural
equivalence
or preorder. For example, in the class of transition systems generated by
ordinary pushdown automata, strong bisimilarity is known to be 
decidable~\cite{PDA:bisimilarity} 
(see also~\cite{Senizergues:TCS2001,  Senizergues:TCS2001_simple, 
Stirling:TCS2001}), but 
no reasonable complexity upper bound is presently known and e.g. 
the simulation
preorder/equivalence on the same class is already 
undecidable~\cite{GH:BPA-other}.
Our main motivation was to investigate whether the picture changes
if we restrict the studied class of systems to visibly pushdown automata,
which are still interesting from the modelling point of view but might
provide more satisfactory decidability/complexity results. Indeed,
our findings confirm the decidability and more reasonable complexity
bounds 
for a number of verification problems on visibly pushdown automata
and their natural subclasses.

We prove EXPTIME-containment of equivalence checking
on visibly pushdown automata (vPDA) for practically all
preorders and equivalences between simulation preorder and
bisimulation equivalence that have been studied in the literature (our focus 
includes simulation, completed simulation, ready simulation,
2-nested simulation and bisimulation). We then study two
natural (and incomparable) subclasses of visibly pushdown
automata: visibly basic process algebra (vBPA) and
visibly one-counter automata (v1CA). In case of
v1CA we demonstrate PSPACE-containment of
the preorder/equivalence checking problems and in case of
vBPA even P-containment. For vBPA we provide also
a direct reduction of the studied problems to equivalence checking
on finite state systems, hence the fast algorithms already
developed for systems with finitely many reachable states
can be directly reused. All the mentioned upper bounds
are matched by the corresponding lower bounds. The PSPACE-hardness proof
for v1CA moreover improves the currently
known DP lower bounds~\cite{JKMS:04} for equivalence checking problems
on ordinary one-counter automata and one-counter nets and some
other problems (see Remark~\ref{rem:1C}).
Finally, we consider the regularity checking problem 
for visibly pushdown automata
and show its P-completeness for vPDA and vBPA, and NL-completeness for v1CA
w.r.t. all equivalences between trace equivalence and bisimilarity.

\emph{Related work.} The main reason why many problems
for visibly pushdown languages become tractable is, as observed
in~\cite{AM:vPDA:04}, that a pair of visibly pushdown automata can be 
synchronized in a similar fashion as finite state automata. We use this
idea to construct, for a given pair of vPDA processes, a single
pushdown automaton where we in a particular way encode the behaviour
of both input processes so that they can alternate in performing
their moves. This is done in such a way that the question of
equality of the input processes w.r.t. a given equivalence/preorder
can be tested by asking about the validity of particular 
modal $\mu$-calculus formulae on the single pushdown process.
A similar result of reducing weak simulation between a pushdown
process and a finite state process (and vice versa) to the model
checking problem appeared in~\cite{KucMay:MFCS:02}. We generalize these
ideas to cover equivalences/preorders between two visibly
pushdown processes and
provide a generic proof for all the equivalence checking
problems. The technical details of our construction are different
from~\cite{KucMay:MFCS:02} and, in particular, our construction
works immediately also for vBPA (as the necessary bookkeeping is
stored in the stack alphabet). As a result we thus show how 
to handle essentially any so far 
studied equivalence/preorder between simulation and bisimulation
in a uniform way for vPDA, vBPA as well as for v1CA.

Regularity problems for deterministic pushdown
automata were studied in~\cite{Stearns:67:IC} and \cite{Valiant:75:JACM}
where a double-exponential algorithm for deciding regularity w.r.t.
language equivalence is given. This decidability result holds also for our
class of visibly pushdown automata (because it is determinizable)
but in our particular 
case we improve the result in two ways: (i) we provide a general
algorithm for regularity checking w.r.t \emph{any} equivalence between trace 
equivalence and bisimilarity, 
and (ii) our algorithm is running in 
polynomial time (and does not require determinization, which is
an expensive operation).

In~\cite{BLS:06} the authors consider language regularity problems for visibly
pushdown automata. 
In particular, they study the question whether a visibly pushdown automaton
is language equivalent to a visibly counter automaton with a given threshold.
In our work we study the regularity problems in the context of the standard
definitions from the concurrency theory, i.e., whether for a given
vPDA process there is a behaviorally equivalent finite state system,
and we consider a wide range of behavioural equivalences, not only the
language equivalence.

The plan of the article is as follows. We introduce the class of
visibly pushdown automata, a range  of behavioral equivalences,
and equivalence checking problems in Section~\ref{sec:definitions}.
The decidability and complexity
of equivalence checking of visibly pushdown automata
and their subclasses is studied in Section~\ref{sec:equivalences}.
Regularity checking problems for the considered classes are
introduced and proved decidable in polynomial time in 
Section~\ref{sec:regularity}. A summary of the results and a
further discussion is presented in Section~\ref{sec:conclusion}.

\section{Definitions} \label{sec:definitions}

A {\em labelled transition system} (LTS) is a triple $(S,\act,\goes{})$
where $S$ is the set of {\em states} (or {\em processes}),
$\act$ is the set of {\em labels} (or {\em actions}), and
$\goes{}\subseteq S\times\act\times S$ \ is the {\em transition
relation}; for each $a\in\act$, we view $\goes{a}$ as a binary relation
on $S$ where
$s \goes{a} s'$ iff $ (s, a, s') \in \goes{}$.
The notation can be naturally extended to $s \goes{w} s'$
for finite sequences of actions $w \in \act^*$. For a process $s \in S$
we define the set of its \emph{initial actions} by $I(s) \defin \{ a \in \act
\mid \exists s' \in S.\ s \goes{a} s' \}$.

We shall now define the studied equivalences/preorders which are between
simulation and bisimilarity. A complete picture of Glabbeek's
linear/branching time hierarchy (spectrum) of behavioral equivalences
is available in~\cite{Glabbeek:PhD,Glabbeek:hierarchy}.
Given $(S,\act,\goes{})$, a binary relation $R \subseteq S\times S$ is a
\begin{enumerate}[$\bullet$]
\item 
\emph{simulation} iff for each $(s,t) \in R$,
$a \in \act$, and $s'$ such that $s \goes{a} s'$
there is $t'$ such that $t \goes{a} t'$ and
$(s',t') \in R$,
\item
\emph{completed simulation} iff $R$ is a simulation
and moreover for each $(s,t) \in R$ it holds that
$I(s) = \emptyset$ if and only if $I(t) = \emptyset$,
\item
\emph{ready simulation} iff $R$ is a simulation
and moreover for each $(s,t) \in R$ it holds that
$I(s) =I(t)$,
\item
\emph{2-nested simulation} iff 
$R$ is a simulation and there is some simulation relation
$R'$ such that $R^{-1} \subseteq R'$, and
\item
\emph{bisimulation} iff $R$ is a simulation
and moreover $R^{-1} = R$.
\end{enumerate}

\begin{figure}[t] 
\begin{center}
\mbox{
\xymatrix@!R@!C@R=5.4ex@C=-11ex{
\text{bisimulation} \ar[d] \\
\text{2-nested simulation} \ar[d] \\
\text{ready simulation} \ar[d] \\
\text{completed simulation} \ar[d] \\
\text{simulation}
}
}
\end{center}
\caption{Hierarchy of simulation-like preorders/equivalences}
\label{figure-spectrum}
\end{figure}
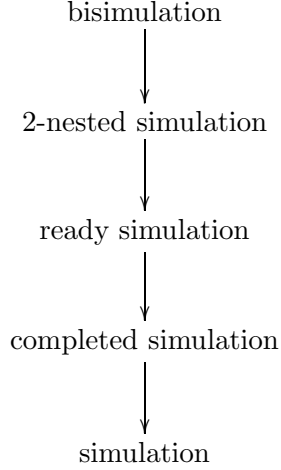

We write 
$s \simul t$ if there is a simulation $R$ such that $(s,t) \in R$, 
$s \csimul t$ if there is a completed simulation $R$ such that $(s,t) \in R$, 
$s \rsimul t$ if there is a ready simulation $R$ such that $(s,t) \in R$, 
$s \tnsimul t$ if there is a 2-nested simulation $R$ such that $(s,t) \in R$, 
$s \sim t$ if there is a bisimulation $R$ such that $(s,t) \in R$.
The relations are called the corresponding \emph{preorders} (except for
bisimilarity, which is already an equivalence). 
For a preorder 
$\sqsubseteq\  \in \{\simul,$ $\csimul,\rsimul,$ $\tnsimul\}$
we define the corresponding equivalence by
$s = t$ iff $s \sqsubseteq t$ and
$t \sqsubseteq s$. We remind the reader of the fact that
$\sim \ \subseteq \ \tnsimul \ \subseteq \ \rsimul \ \subseteq \ \csimul \ \subseteq \
\simul$ and
$\sim \ \subseteq \ \tnsimuleq \ \subseteq \ \rsimuleq \ \subseteq \ \csimuleq \
\subseteq \ \simuleq$ and all inclusions are strict. The hierarchy 
is depicted in Figure~\ref{figure-spectrum}.

We shall use a standard game-theoretic characterization of
(bi)similarity~\cite{Thomas1993TAPSOFT,Stirling:CONCUR95}.
A \emph{bisimulation game} on a pair of processes
$(s_1,t_1)$
is a two-player game between \emph{Attacker} and \emph{Defender}.
The game is played in \emph{rounds}
on pairs of states from $S\times S$.
In each round 
the players change the \emph{current pair of states} $(s,t)$
(initially $s=s_1$ and $t=t_1$)
according to the following rule:
\begin{enumerate}[(1)]
\item
Attacker chooses either $s$ or $t$,
$a \in \act$ and performs a move
$s \goes{a} s'$ or $t \goes{a} t'$.
\item
Defender responds by choosing the opposite
process (either $t$ or $s$) and performs
a move $t \goes{a} t'$ or $s \goes{a} s'$ under the same action $a$.
\item
The pair $(s',t')$ becomes the (new) current pair of states.
\end{enumerate}
A \emph{play} (of the bisimulation game) is a sequence
of pairs of processes formed by the players according to
the rules mentioned above.
A play is finite iff one of the players gets stuck (cannot make
a move); the player who got stuck lost the play and the other player is
the winner.  If the play is infinite then Defender is the winner.

We use the following standard fact.
\begin{prop}
It holds that $s \sim t$ iff Defender has a
winning strategy in the bisimulation game starting with the pair
$(s,t)$, and
$s \not\sim t$ iff Attacker has a
winning strategy in the corresponding game.
\end{prop}

The rules of the bisimulation game can be easily modified in order
to capture the other equivalences/preorders. 

\begin{enumerate}[$\bullet$]
\item In the \emph{simulation preorder game}, Attacker is restricted
to attack only from the (left-hand side) process $s$. In
the \emph{simulation equivalence game}, Attacker can first choose
a side (either $s$ or $t$) but after that he is not allowed to
change the side any more. 
\item
\emph{Completed/ready simulation game}
has the same rules as the simulation game but Defender is
moreover losing in any configuration which breaks the extra
condition imposed by the definition (i.e. $s$ and $t$ should
have the same set of initial actions in case of ready simulation,
and their sets of initial actions should be both empty at the
same time in case of completed simulation). 
\item 
In the
\emph{2-nested simulation preorder game}, Attacker starts playing
from the left-hand side process $s$ and at most once during the play
he is allowed to switch sides (the soundness follows from the characterization
provided in~\cite{AFI:01}). In the \emph{2-nested simulation equivalence game},
Attacker can initially choose any side but he is still restricted that he
can change sides at most once during the play.
\end{enumerate}

We shall now define the model of pushdown automata.
Let $\act$ be a finite set of actions, let
$\Gamma$ be a finite set of stack symbols and let
$Q$ be a finite set of control states. We assume that the sets
$\act$, $\Gamma$ and $Q$ are pairwise disjoint.
A \emph{pushdown automaton} (PDA) over the set of actions $\act$,
stack alphabet $\Gamma$ and control states $Q$ 
is a finite set $\Delta$ of rules of the form
$pX \goes{a} q\alpha$
where $p, q \in Q$, $a \in \act$, $X \in \Gamma$ and
$\alpha \in \Gamma^*$.

A PDA $\Delta$ determines a labelled transition system
$T(\Delta) = (S, \act, \goes{})$ where the states are configurations of
the form state$\times$stack
(i.e. $S = Q\times\Gamma^*$ and configurations like $(p,\alpha)$ are
usually written as $p\alpha$ where the top of the stack $\alpha$
is by agreement on the left)
and the transition relation is determined by the following
prefix rewriting rule.

$$ \pfrule{(pX \goes{a} q\alpha) \in \Delta, \ \ \gamma \in \Gamma^*}
{pX\gamma \goes{a} q\alpha\gamma}$$

In what follows we shall often call the configurations $pX$ as 
(pushdown) \emph{processes} and omit the reference to the
corresponding pushdown automaton and its underlying labelled
transition system whenever it is clear from the context.

A pushdown automaton is called \emph{single-state}
if the set of its control states is a singleton set ($|Q|=1$). In this case
we usually omit the control state from the rules and configurations and
we call such automata
as BPA for \emph{Basic Process Algebra}, which is another standard
terminology.

A pushdown automaton is called 1CA for \emph{one-counter automaton} if the stack
alphabet consists of two symbols only, $\Gamma = \{I,Z\}$, and every
rule is of the form $pI \goes{a} q \alpha$ or
$pZ \goes{a} q \alpha Z$, where $\alpha \in \{I\}^*$.
This means that every configuration reachable from $pZ$ is of the form
$pI^nZ$ where $I^n$ stands for a sequence of $n$ symbols $I$ and $Z$ 
corresponds to
the bottom of the stack (the value zero). We shall simply denote 
such a configuration by $p(n)$ and say that it represents the counter value $n$.

Assume that $\act = \actc \cup \actr \cup \acti$ is partitioned into a
disjoint union of finite sets of call, return and internal actions, respectively.
A \emph{visibly pushdown automaton} (vPDA) is a PDA which
satisfies additional three requirements 
for any rule $pX \goes{a} q\alpha$ 
(where $|\alpha|$ stands
for the length of $\alpha$):
\begin{enumerate}[$\bullet$]
\item if $a \in \actc$ then $|\alpha|=2$ (call),
\item if $a \in \actr$ then $|\alpha|=0$ (return), and
\item if $a \in \acti$ then $|\alpha|=1$ (internal).
\end{enumerate}
Hence in vPDA the type of the input action determines the change in the height
of the stack (call by $+1$, return by $-1$, internal by $0$).  

Visibly basic process algebra (vBPA) and
visibly one-counter automata (v1CA) are defined analogously.

\begin{rem}
For internal actions we allow to modify also the top of the
stack. This model (for vPDA) can be easily seen to be
equivalent to the standard one
(as introduced in~\cite{AM:vPDA:04}) where
the top of the stack does not change under internal actions.
However, when we consider the subclass vBPA, the possibility of 
changing the top of
the stack under internal actions increases the
descriptive power of the formalism. Unlike in~\cite{AM:vPDA:04},
we do not allow to perform return actions on the empty stack.
This mild restriction is essential for our results on regularity
checking in Section~\ref{sec:regularity}
but the results about equivalence checking in 
Section~\ref{sec:equivalences} are valid even without this restriction.
\end{rem}

\begin{exa} Consider the vPDA rules
$pX \goes{a} pXY$, $pX \goes{b} p \epsilon$, $pY \goes{c} p\epsilon$
where $a \in \actc$ and $b,c \in \actr$.
The transition system generated by the root $pX$ looks as follows. \\
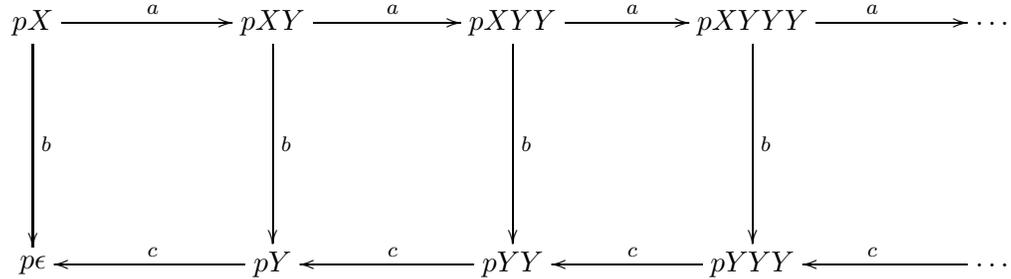
\begin{center}
\mbox{
\xymatrix@!R@!C@R=16ex@C=9ex{
pX \ar[r]^{a} \ar[d]^b & 
pXY \ar[r]^{a} \ar[d]^b &
pXYY \ar[r]^{a} \ar[d]^b  &
pXYYY \ar[r]^{a} \ar[d]^b  & 
\ldots \\
p\epsilon &
pY \ar[l]_c &
pYY \ar[l]_c &
pYYY \ar[l]_c &  \ar[l]_c \ldots
}}
\end{center}

The vPDA process $pX$ (which is in fact also a vBPA process) 
generates an infinite state transition system, which is not
trace equivalent (and hence also not bisimilar) 
to any finite state system. Hence the class of visibly
pushdown processes and visibly BPA processes
strictly contains all finite state processes.
\end{exa}

The question we are interested in is: given a vPDA 
and two of its initial processes $pX$ and $qY$, can we
algorithmically decide whether $pX$ and $qY$ are
related with respect to a given preorder/equivalence
and if yes, what is the complexity?
Similar questions can be asked for vBPA and v1CA.

\begin{rem}
Note that the problem of equivalence checking of two processes belonging
to different visibly pushdown automata (under the same partitioning
of actions) is also covered by the
definition of the problem above. We can always consider 
only a single vPDA by making a disjoint union of the respective 
pushdown automata. 
\end{rem}


\section{Decidability of Preorder/Equivalence Checking} \label{sec:equivalences}

\subsection{Visibly Pushdown Automata}

We shall now study preorder/equivalence checking problems
on visibly pushdown automata. We prove 
their decidability by 
reducing the problems to model checking of an ordinary pushdown system against
a $\mu$-calculus formula.

Let $\Delta$ be a vPDA over the set of actions $\act = \actc \cup \actr 
\cup \acti$, stack alphabet $\Gamma$ and control states $Q$. We shall
construct a PDA $\Delta'$ over the actions
$\act'\defin \act \cup \overline{\act} \cup \{\ell,r\}$
where $\overline{\act} \defin \{ \overline{a} \mid a \in \act\}$, stack alphabet
$\Gamma' \defin G \times G$ where
$G \defin \Gamma \cup (\Gamma\times\Gamma) \cup 
(\Gamma\times\act) \cup \{\epsilon\}$, and
control states $Q' \defin Q \times Q$. 
For notational convenience, 
elements $(X,a) \in \Gamma\times\act$ will be written simply 
as $X_a$.

The idea is that for a given pair of vPDA processes
we shall construct a single PDA process which simulates 
the behaviour of both vPDA processes by repeatedly performing
a move in one of the processes, immediately followed by
a move under the same action in the other process.
The actions $\ell$ and $r$ make it visible, whether the
move is performed on the left-hand side or right-hand side.
The assumption that the given processes are vPDA 
ensures that their stacks are kept synchronized.

We shall define a partial mapping $[\ .\ ,\ . \ ]: \Gamma^*\times\Gamma^* \rightarrow
(\Gamma\times\Gamma)^*$
inductively as follows ($X,Y \in \Gamma$
and $\alpha,\beta \in \Gamma^*$ such that $|\alpha|=|\beta|$):
\begin{center}
\begin{tabular}{lcl}
$[X\alpha,Y\beta]$ & $\defin$ & $(X,Y)[\alpha,\beta]$ \\
$[\epsilon,\epsilon]$ & $\defin$ & $\epsilon$ \ .
\end{tabular}
\end{center}
The mapping provides the possibility to merge stacks.

Assume a given pair of vPDA processes $pX$ and $qY$.
Our aim is to effectively construct a new PDA system $\Delta'$ such that
for every $\bowtie\ \in \{\simul, \simuleq, \csimul, \csimuleq, \rsimul, 
\rsimuleq, \tnsimul, \tnsimuleq, \sim\}$ it is the case that
$pX \bowtie qY$ 
in $\Delta$ if and only if
$(p,q)(X,Y) \models \phi_{\bowtie}$ 
in $\Delta'$ for a fixed $\mu$-calculus formula $\phi_{\bowtie}$.
We refer the reader to~\cite{Kozen:mu-calculus}
for the introduction to the modal $\mu$-calculus.

The set of PDA rules $\Delta'$ is defined as follows.
Whenever $(pX \goes{a} q\alpha) \in \Delta$ then 
the following rules belong to $\Delta'$:
\begin{enumerate}[(1)]
\item 
$(p,p')(X,X') \goes{\ell} (q,p')(\alpha,X'_a)$
for every $p' \in Q$ and $X' \in \Gamma$,
\item 
$(p',p)(X',X) \goes{r} (p',q)(X'_a,\alpha)$ 
for every $p' \in Q$ and $X' \in \Gamma$,
\item $(p',p)(\beta,X_a) \goes{r} (p',q)[\beta,\alpha]$
for every $p' \in Q$ and $\beta \in \Gamma \cup 
(\Gamma\times\Gamma) \cup \{\epsilon\}$,
\item $(p,p')(X_a,\beta) \goes{\ell} (q,p')[\alpha,\beta]$
for every $p' \in Q$ and $\beta \in \Gamma \cup 
(\Gamma\times\Gamma) \cup \{\epsilon\}$,
\item $(p,p')(X,X') \goes{a} (p,p')(X,X')$
for every $p' \in Q$ and $X' \in \Gamma$, and
\item $(p',p)(X',X) \goes{\overline{a}} (p',p)(X',X)$
for every $p' \in Q$ and $X' \in \Gamma$.
\end{enumerate}

From a configuration $(p,q)[\alpha,\beta]$
the rules of the form 1. and 2. select either the left-hand
or right-hand side and perform some transition in the selected
process. The next possible transition (by rules 3. and 4.)
is only from the opposite side
of the configuration than in the previous step.
Symbols of the form $X_a$ where $X \in \Gamma$ and $a \in \act$
are used to make sure that the transitions in these two steps
are due to pushdown rules under the same label $a$. 
Note that in the rules 3. and 4. it is thus 
guaranteed that $|\alpha|= |\beta|$.
Finally, the rules 5. and 6. introduce a number of self-loops
in order to make visible the initial actions of the processes.
Actions currently available in the left-hand side process
are visible as the actions from the set $\act$, while the
actions in the right-hand side process are made visible as the
actions from the set $\overline{\act}$.

\begin{exa}
Consider the vPDA rules: $pX \goes{a} qXY$, $rY \goes{a} sYY$, $rY \goes{b} r$
where $a \in \actc$ and $b \in \actr$. The transition system
generated (in $\Delta'$) by the root $(p,r)(X,Y)$ looks as follows.
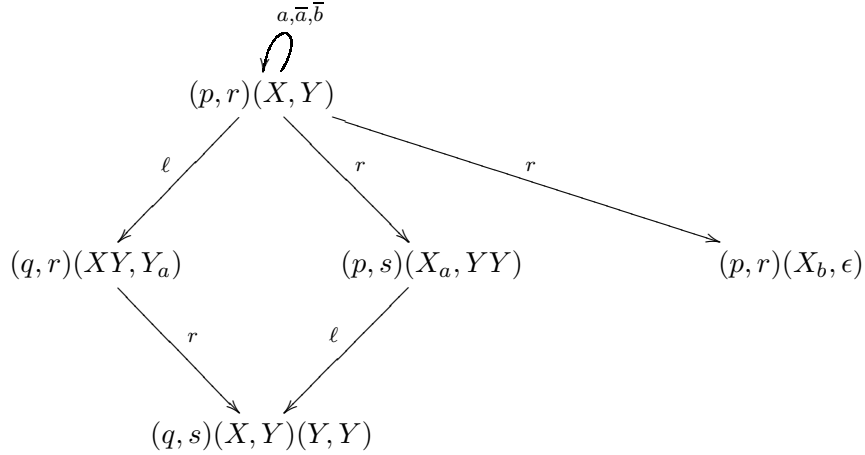
\begin{center}
\mbox{ \xymatrix@R=10ex@C=-4ex{
& (p,r)(X,Y) \ar[dl]_\ell \ar[dr]^r \ar[drrr]^r  
\ar@(ru,u)[]_{a, \overline{a}, \overline{b}}
\\
(q,r)(XY,Y_a) \ar[dr]^r & & (p,s)(X_a,YY) \ar[dl]_\ell & 
\ \ \ \ \ \ \ \ \ \ \ \ \ \ \ \ \ \ \ \ \ \ \ \ \ \ \  & 
(p,r)(X_b,\epsilon) \\
& (q,s)(X,Y)(Y,Y)
}}
\end{center}
\mbox{ }\\[3mm]
Note that the configuration $(p,r)(X_b,\epsilon)$ is stuck because
there is no $b$-labelled transition from the state $pX$.
\end{exa}

\begin{lem} \label{lem:mu}
Let $\Delta$ be a vPDA system over the set of actions $\act$ and $pX$, $qY$
two of its processes. Let $(p,q)(X,Y)$ be
a process in the system $\Delta'$ constructed above.
Let 
\begin{enumerate}[$\bullet$]
\item $\phi_{\simul} \equiv 
\nu Z.\must{\ell}\may{r}Z$,

\item $\phi_{\simuleq} \equiv 
\phi_{\simul} 
 \wedge (\nu Z.\must{r}\may{\ell}Z)$,

\item $\phi_{\csimul} \equiv 
\nu Z.\big(\must{\ell}\may{r}Z \wedge 
(\may{\act}\tr \Leftrightarrow 
\may{\overline{\act}}\tr)\big)$,

\item $\phi_{\csimuleq} \equiv
\phi_{\csimul}
\wedge  
\nu Z.\big(\must{r}\may{\ell}Z \wedge 
(\may{\act}\tr \Leftrightarrow 
\may{\overline{\act}}\tr)\big)$,

\item $\phi_{\rsimul} \equiv 
\nu Z.\big(\must{\ell}\may{r}Z \wedge 
\bigwedge\limits_{a \in \act} (\may{a}\tr \Leftrightarrow \may{\overline{a}}\tr)\big)$,

\item $\phi_{\rsimuleq} \equiv 
\phi_{\rsimul}
\wedge
\nu Z.\big(\must{r}\may{\ell}Z \wedge 
\bigwedge\limits_{a \in \act} (\may{a}\tr \Leftrightarrow \may{\overline{a}}\tr)\big)$,

\item $\phi_{\tnsimul} \equiv 
\nu Z.\big(\must{\ell}\may{r}Z \wedge (\nu Z'.\must{r}\may{\ell}Z')\big)$,

\item $\phi_{\tnsimuleq} \equiv 
\phi_{\tnsimul}
\wedge
\nu Z.\big(\must{r}\may{\ell}Z \wedge (\nu Z'.\must{\ell}\may{r}Z')\big)$, and

\item $\phi_{\sim} \equiv
\nu Z.\must{\ell,r}\may{\ell,r}Z$.
\end{enumerate}
For every  $\bowtie\ \in \{\simul,\simuleq, \csimul, \csimuleq, \rsimul,
\rsimuleq, \tnsimul, \tnsimuleq, \sim\}$
it holds that  
$pX \bowtie qY$ if and only if $(p,q)(X,Y) \models \phi_{\bowtie}$.
\end{lem}
\begin{proof}
We shall argue only for the case of bisimulation. Proofs for the other cases are
similar. The formula $\phi_{\sim}$ 
requires that for every action $\ell$ or $r$ there must follow
at least one of the actions $\ell$ or $r$ such that the same property
holds again. Note that (due to the construction of $\Delta'$),
if the first action
was $\ell$ then necessarily the second one must be
$r$ and vice versa. Finally, $\phi_{\sim}$ is chosen as the
greatest fixed point. The reason is that
it should satisfy also infinite
runs (which have the above mentioned property) as they are
winning for Defender. In what follows we describe formally
the above mentioned intuition. A reader familiar with
recursive formulae of $\mu$-calculus may skip the rest of the proof.

``$\Rightarrow$'': We show that
$p\alpha \sim q\beta$ implies that $(p,q)[\alpha,\beta] \models
\nu Z.\must{\ell,r}\may{\ell,r}Z$. We prove that
the set $F \defin \{ (p,q)[\alpha,\beta] \mid p\alpha \sim q\beta 
\wedge |\alpha|=|\beta| \}$
is a fixed point of the function corresponding to our formula. 
This amounts to checking (without loss of generality as the
other case is symmetric) that for every 
$(p,q)[\alpha,\beta] \in F$ and for every
$\ell$-move
$(p,q)[\alpha,\beta] \goes{\ell} c$ for some configuration
$c$ of $\Delta'$,
there is an $r$-move 
$c \goes{r} (p',q')[\alpha',\beta']$ such that
$(p',q')[\alpha',\beta'] \in F$. As $\nu Z.\must{\ell,r}\may{\ell,r}Z$
is the greatest fixed point, this will establish our claim.
We remind the reader of the fact that there are no $\ell$-transitions
enabled from the configuration $c$.

Let $(p,q)[\alpha,\beta] \in F$ and let
$(p,q)[\alpha,\beta] \goes{\ell} (p',q)(\alpha',Y_a)[\gamma,\delta]$
due to some rule $(pX \goes{a} p'\alpha')\in \Delta$ where
$\alpha=X\gamma$ and $\beta = Y\delta$. This means
that $p\alpha = p X\gamma \goes{a} p'\alpha'\gamma$ and
because $p\alpha \sim q\beta$ we have that
$q\beta=qY\delta \goes{a} q'\beta'\delta$ due to some
rule $(qY \goes{a} q'\beta')\in\Delta$ such that
$p'\alpha'\gamma \sim q'\beta'\delta$.
Hence
$(p',q)(\alpha',Y_a)[\gamma,\delta] \goes{r}
(p',q')[\alpha',\beta'][\gamma,\delta] =
(p',q')[\alpha'\gamma,\beta'\delta]$
and $(p',q')[\alpha'\gamma,\beta'\delta] \in F$ as required.

``$\Leftarrow$'':
We prove that
$(p,q)[\alpha,\beta] \models \nu Z.\must{\ell,r}\may{\ell,r}Z$ implies
that $p\alpha \sim q\beta$. To do so we define a binary relation
$R \defin \{ (p\alpha,q\beta) \mid 
(p,q)[\alpha,\beta] \models \nu Z.\must{\ell,r}\may{\ell,r}Z
\wedge |\alpha|=|\beta| \}$ and
show that $R$ is a bisimulation.
Let $(p\alpha,q\beta) \in R$ and let us without loss of generality 
(the other case is completely symmetric) assume that
$p\alpha \goes{a} p'\alpha'$ due to some rule
$(pX \goes{a} p'\alpha'') \in \Delta$. Hence $\alpha = X\gamma$
and $\alpha' = \alpha''\gamma$ for some $\gamma \in \Gamma^*$.
We aim to show that also $q\beta \goes{a} q'\beta'$ such
that $(p',q')[\alpha',\beta'] \models \nu Z.\must{\ell,r}\may{\ell,r}Z$
and thus $(p'\alpha',q'\beta') \in R$.
The fact that $p\alpha \goes{a} p'\alpha'$ means that in 
$\Delta'$ we have a transition
$(p,q)[\alpha,\beta] = (p,q)[X\gamma,\beta] \goes{\ell}
(p',q)(\alpha'',Y_a)[\gamma,\beta']$ where $\beta =Y\beta'$.
As $(p,q)[\alpha,\beta]$ satisfies the formula
$\nu Z.\must{\ell,r}\may{\ell,r}Z$ (and due to the unfolding
law also the formula 
$\must{\ell,r}\may{\ell,r}(\nu Z.\must{\ell,r}\may{\ell,r}Z)$),
we have that
$(p',q)(\alpha'',Y_a)[\gamma,\beta'] \goes{r}
(p',q')[\alpha'',\beta''][\gamma,\beta']$ due to some
rule $(qY \goes{a} q'\beta'') \in \Delta$ such that
$(p',q')[\alpha'',\beta''][\gamma,\beta'] \models
\nu Z.\must{\ell,r}\may{\ell,r}Z$. Note that from
$(p',q)(\alpha'',Y_a)[\gamma,\beta']$ no $r$-action
is enabled on the left-hand side
and there is no $\ell$-action available, so this is indeed 
the only possibility.
This, however, means that $(p',q')[\alpha'',\beta''][\gamma,\beta'] =
(p',q')[\alpha''\gamma,\beta''\beta'] =
(p',q')[\alpha',\beta''\beta']$, and 
$q\beta \goes{a} q'\beta''\beta'$, which implies
that $(p'\alpha',q'\beta''\beta') \in R$.
\end{proof}

\begin{thm} \label{thm:EXP-vPDA}
Simulation, completed simulation, ready simulation and
2-nested simulation  preorders and equivalences, as well as
bisimulation equivalence are decidable on vPDA and all these
problems are EXPTIME-complete.
\end{thm}
\begin{proof}
EXPTIME-hardness (for all relations between simulation
preorder and bisimulation equivalence)  
follows from Theorem 6 in~\cite{KucMay:MFCS:02}.
There the authors give a reduction from the acceptance problem of
alternating linear bounded automaton (LBA) to bisimilarity
checking on pushdown automata. 
For a given alternating
LBA they construct two PDA processes that have
identical stack contents (and hence also their heights) 
during any bisimulation game played on them. These processes
are bisimilar if and only if the given LBA does not accept
the input string. The processes are constructed
in such a way that Attacker can nondeterministically push configurations
of the LBA, symbol by symbol, onto the stack of the first process.
Defender has no other choice than to mimic this behaviour 
in the second process. If
the current control state of the LBA is existential, then
it is Attacker who selects the next control state. On the other hand, if
the current control state is universal, Defender is selecting the
next control state. This part of the construction
uses the so-called \emph{Defender's Forcing 
Technique}~\cite{JS:08:JACM}, an instance
of which is also explained in the proof of Lemma~\ref{lem:pspace}.
Attacker wins the game iff a configuration containing an accept
control state is ever pushed onto the stacks. 

Of course, Attacker can also decide 
to ``cheat'' and push symbols that do not in fact represent a valid
computation of the LBA. To prevent this, a checking phase is added
to the constructed processes. Defender can at any point decide to enter
this phase and what will happen is that the players will pop exactly
$n+4$ symbols (a control state is a part of the configuration)
 from the stack of both processes, where $n$ is the length
of the input. During this phase it is verified whether the three top-most
symbols on the stack match (according to the LBA computation step)
with the corresponding three symbols in the
previous configuration.
If not, Defender wins the game,
otherwise Attacker is the winner. This means that Defender cannot gain
anything by entering the checking phase as long as Attacker follows
a valid computation of the LBA, but also that Attacker cannot gain anything
by ``cheating'' as this can be immediately punished by Defender.

Though conceptually elegant, the technical details of the reduction are 
rather tedious. A full construction
can be found in~\cite{KM:RS-2002-01} (Theorem 3.3), however,
already from the presented proof idea it is easy to see that
the pushdown processes in the proof can be constructed as visibly
pushdown processes because the heights of their stacks are synchronized. 
Note also that, even though Theorem 3.3 in~\cite{KM:RS-2002-01} is 
formulated only for bisimulation,
Attacker's winning strategy is given by playing only in the left-hand side 
process and hence the EXPTIME-hardness holds for
any relation between simulation preorder and bisimulation equivalence. 
 
For the containment in EXPTIME observe that all our equivalence
checking problems were reduced in polynomial time to
model checking of a pushdown automaton against a 
formula of modal $\mu$-calculus with alternation depth $1$. 
The complexity of the
model checking problem for a pushdown automaton with $m$ states
and $k$ stack symbols and a formula of the size $n_1$ and of the
alternation depth $n_2$ is $O( (k2^{cmn_1n_2})^{n_2}) )$
for some constant $c$~\cite{Wal:PDA:01}. In our case
for a given vPDA system with $m$ states and $k$ stack
symbols we construct a PDA system with $O(m^2)$ states
and with $O(k^3\cdot|\act|)$ stack symbols (used in the transition rules). 
The size of the $\mu$-calculus formula is $O(|\act|)$ (for some
equivalences even $O(1)$) and the alternation depth is $1$.
Hence the overall time complexity of checking whether
two vPDA processes $pX$ and $qY$ are 
equivalent is $(k^3\cdot|\act|)2^{O(m^2\cdot|\act|)}$.
\end{proof}


\subsection{Visibly Basic Process Algebra}

We shall now focus on the complexity of preorder/equivalence
checking on vBPA, a strict subclass of vPDA.

\begin{thm}
Simulation, completed simulation, ready simulation and
2-nested simulation  preorders and equivalences, as well as
bisimulation equivalence are P-complete on vBPA.
\end{thm}
\begin{proof}
We can w.l.o.g. assume that the vBPA processes tested for
a given equivalence have only three actions.
This is due to Theorem~4 in~\cite{Srba:concur:2001} where
a polynomial time reduction from bisimulation checking of
two BPA processes to bisimulation checking of two BPA processes
over a singleton action set is given. In order to make the resulting
processes visibly BPA, we modify in a straightforward way 
the reduction in~\cite{Srba:concur:2001} by using three actions
so that we can distinguish between push, pop and internal actions.

Now by using the arguments from the proof of
Theorem~\ref{thm:EXP-vPDA}, the complexity of equivalence checking
on vBPA with $|\act|=3$
is therefore $O(k^3)$ where
$k$ is the cardinality of the stack alphabet (of the
reduced processes according to~\cite{Srba:concur:2001}) and where $m=1$. 

P-hardness of the equivalence checking problems
was proved in~\cite{SawaJancar:PTIME-hard:CaI:2005} 
even for finite state systems.
\end{proof}

In fact, for vBPA we can introduce even tighter complexity
upper bounds by reducing it to preorder/equivalence checking
on finite state systems. These results are due to the following
decomposition property.

\begin{lem} \label{lem:decomposition}
Let $\Delta$ be a vBPA system.
We have $X\alpha \sim X'\alpha'$ in $\Delta$
(where $X,X' \in \Gamma$ and $\alpha,\alpha'\in \Gamma^*$)
if and only if
\begin{enumerate}[\em(i)]
\item $X \sim X'$ in $\Delta$, and
\item if ($X \goes{}^* \epsilon$ or $X' \goes{}^* \epsilon$) then
$\alpha \sim \alpha'$ in $\Delta$.
\end{enumerate}
\end{lem}
\begin{proof}
The validity of this claim can be easily seen by using the
game characterization. Obviously, in \emph{any} play of the game
starting from $X\alpha$ and $X'\alpha'$,
if the players ever reach either the process $\alpha$ or $\alpha'$,
then it happens simultaneously in one round of the game
and the players continue from the pair $\alpha$ and $\alpha'$.
This means that conditions (i) and (ii) imply that $X\alpha \sim X'\alpha'$.
On the other hand, if (i) does not hold then Attacker can win
by playing only from $X$ and $X'$, never touching $\alpha$ and $\alpha'$
during the game. If (ii) does not hold (i.e., $X$ or $X'$ can be
removed from the stack and $\alpha \not\sim \alpha'$),
then Attacker wins first
by playing a sequence $w$ from $X$ or $X'$ such that $X \goes{w} \epsilon$
or $X' \goes{w} \epsilon$. If Defender could follow this sequence
in the other process then the players reach the pair $\alpha$ and $\alpha'$
and Attacker wins because $\alpha \not\sim \alpha'$. 
\end{proof}

\begin{thm} \label{thm:vBPAreduction}
Simulation, completed simulation, ready simulation and
2-nested simulation  preorders and equivalences, as well as
bisimulation equivalence on vBPA are reducible to checking 
the same preorder/equivalence on finite state systems.
For any vBPA process $\Delta$ 
(with the natural requirement that every stack symbol appears 
at least in one rule from $\Delta$), the reduction is computable in time
$O(|\Delta|)$ and outputs a finite state system with $O(|\Delta|)$
states and $O(|\Delta|)$ transitions.
\end{thm}
\begin{proof}
Let $\act = \actc \cup \actr \cup \acti$ be the set of
actions and let $\Gamma$ be the stack alphabet of a given
vBPA system $\Delta$ (we shall omit writing the control
states as this is a singleton set). 
Let $S \defin \{ (Y,Z) \in \Gamma\times\Gamma \mid 
(X \goes{a} YZ) \in \Delta \text{ for some $X \in \Gamma$
and $a \in \actc$ }\}$.
We construct
a finite state transition system
$T = (\Gamma \cup \{\epsilon\} \cup S,
\act \cup \{1,2\}, \goesw{})$
with new fresh actions $1$ and $2$ as follows.
For every vBPA rule $(X \goes{a} \alpha) \in \Delta$, we add the
transitions: 
\begin{enumerate}[$\bullet$]
\item $X \goesw{a} \epsilon$ if $a \in \actr$ (and $\alpha = \epsilon$),
\item $X \goesw{a} Y$ if $a \in \acti$ and $\alpha = Y$, 
\item $X \goesw{a} (Y,Z)$ if $a \in \actc$ and $\alpha = YZ$,
\item $(Y,Z) \goesw{1} Y$ if $a \in \actc$ and $\alpha = YZ$, and
\item $(Y,Z) \goesw{2} Z$ if $a \in \actc$ and $\alpha = YZ$
such that $Y \goes{}^* \epsilon$.
\end{enumerate}

See Figure~\ref{fig:transformation} for an example
of the transformation. 
The intuition is that all transitions in $\Delta$ that do not push any
extra symbols on the stack are directly mimicked in $T$.
For the push transitions we know, thanks to Lemma~\ref{lem:decomposition},
that if Attacker has a winning strategy, then it is either
by attacking only from the top-most symbols pushed on the stacks,
or from the rest of the stacks (if the top-most symbols can be eventually
popped). Hence Attacker has the possibility to choose one of these
options by playing the actions $1$ or $2$ and the transition system $T$
remains so finite state.
 
Note that the set $\{ Y \in \Gamma \mid Y \goes{}^* \epsilon\}$ 
can be (by standard techniques) computed in time $O(|\Delta|)$.
Moreover, the finite state system $T$ has $O(|\Delta|)$ states
and $O(|\Delta|)$ transitions.

\begin{figure}[t]
\begin{center}
\begin{tabular}{cc}
\begin{minipage}{4cm}
$X \goes{a} Y$ \\ $X \goes{b} \epsilon$ \\ $X \goes{c} XY$ \\ 
$Y \goes{b} \epsilon$
\end{minipage}
&
\begin{minipage}{4cm}
\mbox{
\xymatrix@!R@!C@R=14ex@C=14ex{
X \ar@/^7pt/@{=>}[d]^c \ar@{=>}[r]^b \ar@{=>}[dr]^a & \epsilon \\
(X,Y) \ar@/^7pt/@{=>}[u]^1 \ar@{=>}[r]^2 & Y \ar@{=>}[u]_b
}}
\end{minipage}
\end{tabular}
\end{center}
\caption{Transformation of a vBPA into a finite state system} 
\label{fig:transformation}
\end{figure}
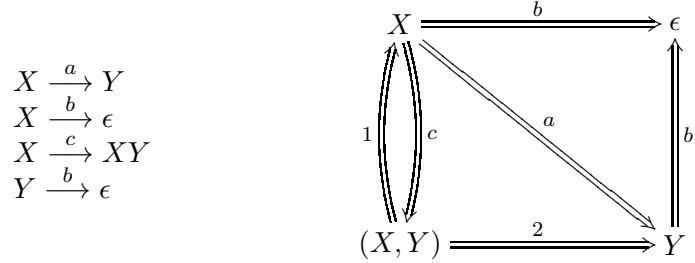

Because of the decomposition property proved in Lemma~\ref{lem:decomposition}, 
for any $X, Y \in \Gamma$ we have that
$X \sim Y$ in $\Delta$ iff 
$X \sim Y$ in $T$.
It is easy to check that the fact above holds also for any 
other preorder/equivalence as stated by the theorem.
\end{proof}

This means that for preorder/equivalence checking on vBPA we
can use the efficient algorithms already developed for
finite state systems.
For example, for finite state transition systems with
$k$ states and $t$ transitions, bisimilarity
can be decided in time $O(t \log k)$~\cite{PaigeTarjan87}.
Hence bisimilarity on a vBPA system $\Delta$ is decidable in
time $O(|\Delta|\cdot\log |\Delta|)$.

\subsection{Visibly One-Counter Automata}

We will now continue studying preorder/equi\-va\-lence
checking problems on v1CA, a strict subclass of vPDA
and an incomparable class with vBPA (w.r.t. bisimilarity).
We start by showing PSPACE-hardness of the problems.
The proof is by reduction from a PSPACE-complete problem
of emptiness of one-way alternating finite automata over
a one-letter alphabet~\cite{Holzer:95}.

A \emph{one-way alternating finite automaton over a one-letter
alphabet} is a 5-tuple $A = (\Qe, \Qu, q_0, \delta, F)$
where $\Qe$ and $\Qu$ are finite and disjoint sets of
existential, resp. universal control states, $q_0 \in 
\Qe \cup \Qu$ is the initial state, $F \subseteq \Qe \cup \Qu$
is the set of final states and $\delta: \Qe \cup \Qu \rightarrow
2^{\Qe \cup \Qu}$ is the transition function. We assume without
loss of generality that $|\delta(q)|>0$ for all $q \in \Qe \cup \Qu$.

A \emph{computation tree} for an input word of the form
$I^n$ (where $n$ is a natural number and
$I$ is the only letter in the input alphabet)
is a tree where every branch has exactly $n+1$ nodes labelled by 
control states from $\Qe \cup \Qu$ such that the root is labelled with $q_0$
and every non-leaf node that is already labelled by some 
$q \in \Qe \cup \Qu$ such that $\delta(q) = \{q_1, \ldots, q_k\}$ has either
\begin{enumerate}[$\bullet$]
\item one child labelled by $q_i$ for some
 $i$, $1 \leq i \leq k$, if $q \in \Qe$, or
\item $k$ children labelled by $q_1$, \ldots, $q_k$, if $q \in \Qu$.
\end{enumerate}
A computation tree is \emph{accepting} if the labels of all its
leaves are final (i.e. belong to $F$). The language of $A$ is defined
by $L(A) \defin \{I^n \mid I^n \text{ has some accepting computation
tree }\}$.

The emptiness problem for one-way alternating finite automata over a 
one-letter alphabet (denoted as \Empty) is to decide whether $L(A) = \emptyset$
for a given automaton $A$. The problem $\Empty$ is known to be
PSPACE-complete due to Holzer~\cite{Holzer:95} and it is proved by 
a series of reductions (see also a direct
proof by Jan\v{c}ar and Sawa~\cite{JS:1L-AFA}).

In what follows we shall demonstrate a polynomial time reduction from
$\Empty$ to equivalence/preorder checking on visibly one-counter automata.
We shall moreover show the reduction for any (arbitrary) relation
between simulation preorder and bisimulation equivalence. This 
in particular covers all preorders/equivalences introduced in
this article. 

\begin{lem} \label{lem:pspace}
All relations between simulation preorder and bisimulation equivalence
are PSPACE-hard on v1CA. 
\end{lem}
\begin{proof}
Let $A=(\Qe,\Qu,q_0,\delta,F)$ be a given instance of $\Empty$. 
We shall construct a visibly one-counter automaton $\Delta$
over the set of actions
$\actc \defin \{i\}$, 
$\actr \defin \{d_q \mid q \in \Qe \cup \Qu \}$, 
$\acti \defin \{a, e\}$ 
and with control states
$Q \defin \{p,p',t\} \cup \{q, q', t_q \mid q \in \Qe \cup \Qu \}$
such that
\begin{enumerate}[$\bullet$]
\item if $L(A)=\emptyset$ then Defender has a winning
strategy from $pZ$ and $p'Z$ in the bisimulation game
(i.e. $pZ \sim p'Z$), and
\item if $L(A)\not=\emptyset$ then Attacker has
a winning strategy from $pZ$ and $p'Z$ in the simulation
preorder game (i.e. $pZ \not\simul p'Z$).
\end{enumerate}
The intuition is that Attacker generates some counter value $n$
in both of the processes $pZ$ and $p'Z$
and then switches into a checking phase by changing the
configurations to $q_0(n)$ and $q'_0(n)$. Now the players 
decrease the counter and change the control states according
to the function $\delta$. Attacker selects the successor
in any existential configuration, while Defender makes the choice
of the successor in every
universal configuration. Attacker wins if the players reach
a pair of configurations $q(0)$ and $q'(0)$ where $q \in F$.

We shall now define the set of rules $\Delta$.
The initial rules allow Attacker (by performing repeatedly the
action $i$) to set the counter into an
arbitrary number, i.e., Attacker generates a candidate word from $L(A)$.
\begin{center}
\begin{tabular}{ll}
$pZ \goes{i} pIZ$ \ \ \ \ \ \ \  & $p'Z \goes{i} p'IZ$ \\
$pI \goes{i} pII$ & $p'I \goes{i} p'II$ \\
$pZ \goes{a} q_0Z$ & $p'Z \goes{a} q'_0Z$ \\
$pI \goes{a} q_0I$ & $p'I \goes{a} q'_0I$ 
\end{tabular}
\end{center}
Observe that Attacker is at some point forced to perform
the action $a$ (an infinite play is winning for Defender)
and switch to the checking phase starting from
$q_0(n)$ and $q'_0(n)$.

Now for every existential state $q \in \Qe$ with
$\delta(q) = \{q_1, \ldots, q_k\}$ and for
every $i \in \{1, \ldots, k\}$ we add the following rules.
\begin{center}
\begin{tabular}{ll}
$qI \goes{d_{q_i}} q_i$ \ \ \ \ \ \ \ &
$q'I \goes{d_{q_i}} q'_i$ 
\end{tabular}
\end{center}
This means that Attacker can decide on the successor $q_i$ of $q$ and
the players in one round move from the pair $q(n)$ and $q'(n)$ into
$q_i(n-1)$ and $q'_i(n-1)$.

Next for every universal state $q \in \Qu$ with
$\delta(q) = \{q_1, \ldots, q_k\}$ and for
every $i \in \{1, \ldots, k\}$ 
we add the rules
\begin{center}
\begin{tabular}{ll}
$qI \goes{a} tI$ \ \ \ \ \ \ \ & $q'I \goes{a} t_{q_i}I$  \\
$qI \goes{a} t_{q_i}I$ 
\end{tabular}
\end{center}
and for every $q,r \in \Qe \cup \Qu$ such that $q \not= r$ we add
\begin{center}
\begin{tabular}{ll}
$tI \goes{d_{q}} q$ \ \ \ \ \ \ \ &
$t_{q}I \goes{d_{q}} q'$ \\
& $t_{q}I \goes{d_{r}} r$ \ . 
\end{tabular}
\end{center}

These rules are more complex and they correspond to a particular
implementation of the so-called \emph{Defender's Forcing Technique}
(for further examples see e.g.~\cite{JS:08:JACM}). We shall explain
the idea by using Figure~\ref{figure-choice}. 
\begin{figure}
\begin{center}
\mbox{
\xymatrix@!R@!C@R=8.5ex@C=3ex{
q(n) \ar[dd]^a \ar[ddrr]^a \ar[ddrrr]^a \ar[ddrrrr]^a
& & & q'(n) \ar[ddl]_a \ar[ddr]^a \ar[dd]^a \\ \\ 
t(n)\ar[dd]^{d_{q_i}} \ar[ddrr]^{d_{r}}
& & t_{q_1}(n) & t_{q_i}(n) \ar[ddl]_{d_{r}} \ar[dd]^{d_{q_i}} 
& t_{q_k}(n) \\ & & \forall r \not= q_i \mbox{\ \ \ \ \ \ \ } \\
q_i(n-1) & & r(n-1) & q'_i(n-1)
}}
\end{center}
\caption{Defender's Choice: $q\in\Qu$ and $\delta(q)=\{q_1,\ldots,q_k\}$}
\label{figure-choice}
\end{figure}
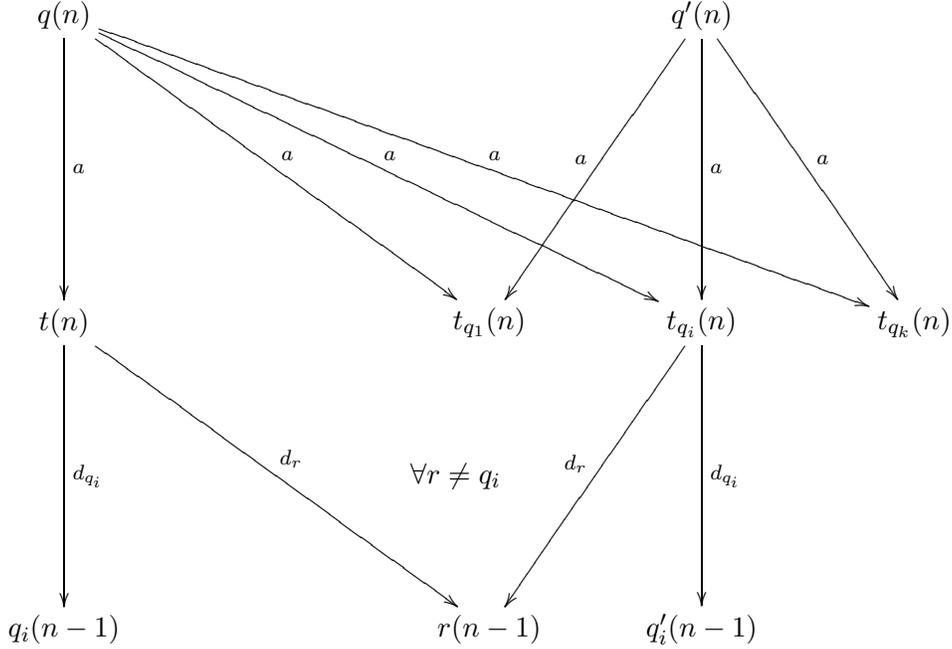
Assume that $q \in \Qu$ and $\delta(q)=\{q_1, \ldots, q_k\}$.
In the first round of the bisimulation game
starting from $q(n)$ and $q'(n)$ where $n>0$,
Attacker is forced
to take the move $q(n) \goes{a} t(n)$. On any other move Defender
answers by immediately reaching a pair of syntactically equal
processes (and thus wins
the game). Defender's answer on Attacker's move $q(n) \goes{a} t(n)$
is to perform $q'(n) \goes{a} t_{q_i}(n)$ for some $i \in \{1, \ldots, k\}$.
The second round thus starts from the pair $t(n)$ and $t_{q_i}(n)$.
Should Attacker choose to play the action $d_{r}$ for
some state $r$ such that $r \not= q_i$ (on either side), Defender
can again reach a syntactic equality and win. Hence Attacker is
forced to play the action $d_{q_i}$ on which Defender answers by
the same action in the opposite process and the players reach
the pair $q_i(n-1)$ and $q'_i(n-1)$. Note that it was Defender
who selected the new control state $q_i$.

Finally, for every $q \in F$ we add the rule
$$qZ \goes{e} qZ\ .$$

It is easy to see that $\Delta$ is a visibly one-counter
automaton and we shall now argue for the correctness of the reduction.

Assume that $L(A)=\emptyset$. We shall argue that Defender
has a winning strategy in the bisimulation game starting
from $pZ$ and $p'Z$. In the first phase Attacker can generate an
arbitrary number of the symbols $I$ on the stacks. At some point
he has to perform the action $a$ 
and switch to $q_0(n)$ and $q'_0(n)$ for some $n$. The
players now remove the symbols $I$ one by one and change the control
states according to the function $\delta$. As $L(A)=\emptyset$,
we know that no computation tree can be accepting. This 
means that whatever choices Attacker makes in existential states,
Defender can still select suitable successors of universal states
such that when the players empty the whole counter and arrive to the 
pair $qZ$ and $q'Z$, Defender guarantees that $q \not\in F$.
Therefore $qZ$ and $q'Z$ are stuck and thus Defender has a 
winning strategy in the bisimulation game.

On the other hand, if $L(A)\not=\emptyset$, we will demonstrate
Attacker's winning strategy in the simulation preorder game
starting from $pZ$ and $p'Z$. Attacker first forces (by repeatedly
performing the action $i$ followed by one action $a$) to reach
a pair of states $q_0(n)$ and $q_0'(n)$ such that $I^n \in L(A)$.
In the checking phase, there is an accepting computation tree for
the word $I^n$ and hence 
Attacker can make existential choices such that whatever universal
choices Defender makes, the players arrive to the situation
$qZ$ and $q'Z$ for some $q \in F$. Now Attacker wins by playing
$qZ \goes{e} qZ$ to which Defender has no answer. Notice that
during the whole game (and particularly during the part where Defender
chooses a successor of a universal state) Attacker can make his moves
only on the left-hand side. Therefore $pZ \not\simul p'Z$ as required.
\end{proof}

\begin{rem} \label{rem:1C}
The reduction above works
also for a strict subclass of one-counter automata called
one-counter nets (where it is not allowed to
test for zero, see e.g.~\cite{JKMS:04}). It is enough to replace
the final rule $qZ \goes{e} qZ$ 
with two new rules $q \goes{e} q$ and $q'I \goes{e} q' I$
for every $q \in F$. Moreover, a slight modification of the
system allows to show PSPACE-hardness of simulation preorder
checking between one-counter automata and finite state systems
and vice versa. Hence the previously known DP lower bounds~\cite{JKMS:04} for 
all relations between simulation preorder and bisimulation equivalence
on one-counter nets (and one-counter automata) as well as
of simulation preorder/equivalence between one-counter automata
and finite state systems, and between finite state systems and
one-counter automata are raised to PSPACE-hardness.
\end{rem}

We are now ready to state the precise complexity of
(bi)simulation-like preorders/equi\-va\-lences on visibly one-counter
automata.

\begin{thm} 
Simulation, completed simulation, ready simulation and
2-nested simulation  preorders and equivalences, as well as
bisimulation equivalence are PSPACE-complete on v1CA.
\end{thm}
\begin{proof}
PSPACE-hardness follows from Lemma~\ref{lem:pspace}.
Containment in PSPACE is due to Lemma~\ref{lem:mu} and
due to~\cite{Serre:06} where
it was very recently showed that model checking modal $\mu$-calculus 
on one-counter automata is decidable in PSPACE.
The only slight complication is that the system
used in Lemma~\ref{lem:mu} is not necessarily 
a one-counter automaton. All stack symbols are of the
form $(I,I)$ or $(Z,Z)$ which is fine, except for the
very top of the stack where more different stack symbols are used.
Nevertheless, by standard techniques, the top of the stack
can be remembered in the control states in order to apply
the result from~\cite{Serre:06}.
\end{proof}

\section{Decidability of Regularity Checking} \label{sec:regularity}

In this section we ask the question whether a given
vPDA process is equivalent to some finite state system. Should this be
the case, we call the process \emph{regular}
(w.r.t. the considered equivalence). 

\begin{defi}
A vPDA process $pX$ is \emph{regular w.r.t. a given equivalence $\equiv$}
iff there is a finite state process $F$ such that 
$pX \equiv F$.
\end{defi}

Note that we do not fix any
particular equivalence $\equiv$ in the definition above.
This is on purpose as the results in this section are generic
and hold for regularity w.r.t. many different equivalence notions.

We shall now give a semantical characterization
of regular vPDA processes via the property of
unbounded popping and a polynomial time
decision algorithm to test whether a given process satisfies this property. 

Let $\act = \actc \cup \actr \cup \acti$ be the set of
actions of a given vPDA. We define a function 
$h: \act \rightarrow \{-1,0,+1\}$ by
$h(a) = +1$ for all $a \in \actc$,
$h(a) = -1$ for all $a \in \actr$, and
$h(a) = 0$ for all $a \in \acti$. The function
$h$ can be naturally extended to sequences of actions by
$h(a_1\ldots a_n)= \sum_{i \in \{1,\ldots,n\}}h(a_i)$. 
Observe now that
for any computation $p\alpha \goes{w} q\beta$ we have
$|\beta| = |\alpha| + h(w)$.

\begin{defi}
Let $pX$ be a vPDA process. We say that
$pX$ provides \emph{unbounded popping} if for
every natural number $d$ there is a process 
$q\beta$ and a word $w \in \act^*$ such that
$h(w) \leq -d$ and $pX \goes{}^* q\beta \goes{w}$\ .
\end{defi}

\begin{lem} \label{lem:reg1}
Let $pX$ be a vPDA process which provides unbounded
popping. Then $pX$ is not regular w.r.t. trace equivalence.
\end{lem}
\begin{proof}
By contradiction. Let $pX$ 
be trace equivalent to some finite state system $F$ with $n$
states. Let us consider a trace $w_1w_2$ such that
$pX \goes{w_1} q\beta \goes{w_2}$ for some $q\beta$ and
$h(w_2) \leq -n$. Such a trace must exist because $pX$
provides unbounded popping. The trace $w_1w_2$ must be
executable also in $F$. However, because $F$ has $n$ states,
during the computation on $w_2$, it must necessarily enter
twice some state $q_{rep}$. 
Because $q\beta$ decreases the stack height by more than $n$
symbols during the computation on $w_2$, we can moreover
assume that the second occurrence of $q_{rep}$ in $F$ happened
when the corresponding pushdown configuration
had the stack height strictly smaller than
in the first occurrence of $q_{rep}$.
This means that such a computation on $w_2$ forms a loop on
a substring $w'$ of $w_2$ which was observed between
the first and the second occurrence of $q_{rep}$ and
$h(w') < 0$.

By repeating the loop under $w'$ (where $w_2=xw'y$) 
in $F$ sufficiently many times ($|w_1|+|x|+|y|+2$ times is surely enough),
$F$ can achieve a trace $w=w_1xw'^{|w_1|+|x|+|y|+2}y$ with $h(w) < -1$.
However, this trace is not possible from $pX$
(any word $w$ such that $pX \goes{w}$ satisfies that $h(w) \geq -1$).
This is a contradiction.
\end{proof}

\begin{lem} \label{lem:reg2}
Let $pX$ be a vPDA process which does not provide unbounded
popping. Then $pX$ is regular w.r.t. bisimilarity.
\end{lem}
\begin{proof}
Assume that $pX$ does not provide unbounded popping. In other
words, there is a constant $d_{\mathit{max}}$ such that
for every process $q\beta$ reachable from $pX$ it is the
case that for any computation starting from $q\beta$,
the stack height $|\beta|$ cannot be decreased by more than
$d_{\mathit{max}}$ symbols. This means that in any reachable
configuration it is sufficient 
to remember only $d_{\mathit{max}}$ top-most stack symbols
and hence the system can be up to bisimilarity 
described as a finite state
system (in general of exponential size).
\end{proof}

\begin{thm} \label{thm:popping}
Let $pX$ be a vPDA process. Then,
for any equivalence relation 
between trace equivalence and bisimilarity,
$pX$ provides unbounded popping if and only if $pX$ is not regular.
\end{thm}
\begin{proof}
Directly from Lemma~\ref{lem:reg1} and Lemma~\ref{lem:reg2}.
\end{proof}

We shall now show that unbounded popping property is
decidable in polynomial time and we also
take a closer look at the exact complexity of regularity
checking problems on vPDA, vBPA and v1CA. The results
are presented in the following three lemmas. 

\begin{lem} \label{lem:inP}
Regularity checking of vPDA w.r.t. any equivalence
between trace equivalence and bisimilarity
is decidable in deterministic polynomial time.
\end{lem}
\begin{proof}
We can check, for every $q \in Q$ and $Y \in \Gamma$,
whether the regular set
$\post^*(qY) \cap \pre^*(\{r\epsilon \mid r \in Q\})$ is infinite.
A nondeterministic finite automaton recognizing this language
can be constructed in polynomial time because
$\pre^*$ and $\post^*$ preserve regularity and
are polynomial time computable 
(see e.g.~\cite{CONCUR::BouajjaniEM1997,EHRS00}),
and the check whether the resulting automaton has an infinite
language amounts to searching (in polynomial time)
for a reachable cycle from which there is a path to some accept state.
Observe now that if the above mentioned regular language is infinite, 
then $qY$ has infinitely many different successors
with higher and higher stacks such that all of them can
be completely emptied.
To see whether a given vPDA process $pX$ provides unbounded
popping (and hence it is nonregular due to Theorem~\ref{thm:popping}),
it is now enough to test whether
$pX \in \pre^*(qY\Gamma^*)$ for some $qY$ satisfying
the condition above. The test can be again done in polynomial 
time~\cite{CONCUR::BouajjaniEM1997,EHRS00}.
\end{proof}

\begin{lem} \label{lem:Phard}
Regularity checking of vBPA w.r.t. any equivalence
between trace equivalence and bisimilarity
is P-hard.
\end{lem}
\begin{proof}
In order to argue that regularity for 
vBPA is P-hard we first consider the following problem.
Let $\Delta$ be a BPA system over the set of actions $\act$
and a stack alphabet $\Gamma$. Let $X \in \Gamma$.
The language of $X$ recognized by the empty stack is defined as 
$L(X) = \{ w \in \act^* \mid X \goes{w} \epsilon \}$. It is known that
the problem whether $L(X)=\emptyset$ (which we shall call BPA-EMPTY)
is P-hard, even under the assumption that there are only finitely
many configurations reachable from the process $X$, and that every rule
$Y \goes{a} \alpha$ in $\Delta$ satisfies that $|\alpha| \leq 2$.
P-hardness of BPA-EMPTY follows from a simple logarithmic
space reduction from the Monotone Circuit Value problem
(see~\cite[p. 177]{GHR:Pcompleteness} for details). We shall reduce BPA-EMPTY
to regularity checking on vBPA. Let $\Delta$ together with a stack 
symbol $X \in \Gamma$, 
which has finitely many reachable states,
be a given instance of BPA-EMPTY. We construct (in logarithmic space)
a vBPA system $\Delta'$ over the partitioned action alphabet 
$\actc=\{c\}$, $\actr=\{r\}$, $\acti=\{i,e\}$ and 
the stack alphabet $\Gamma'= \Gamma \cup \{X',B,C,D\}$, where
$X'$, $B$, $C$ and $D$ are fresh stack symbols, such that 
\begin{enumerate}[$\bullet$]
\item if $L(X)=\emptyset$ in $\Delta$ then $X'$ is a regular process
in $\Delta'$ w.r.t. bisimilarity, and 
\item if $L(X)\not=\emptyset$ in $\Delta$ then
$X'$ is not a regular process in $\Delta'$ w.r.t. trace equivalence. 
\end{enumerate}
We build $\Delta'$ from $\Delta$ as follows: 
\begin{enumerate}[$\bullet$]
\item for every $(Y \goes{a} \alpha) \in \Delta$ where $|\alpha|=2$
we add to $\Delta'$ the rule $Y \goes{c} \alpha$,
\item for every $(Y \goes{a} \alpha) \in \Delta$ where $|\alpha|=0$
we add to $\Delta'$ the rule $Y \goes{r} \alpha$, and
\item for every $(Y \goes{a} \alpha) \in \Delta$ where $|\alpha|=1$
we add to $\Delta'$ the rule $Y \goes{i} \alpha$.
\end{enumerate}
This does not change the answer to the emptiness problem and 
the system $\Delta'$ becomes
visibly BPA. If we now add the following rules to $\Delta'$
$$
X' \goes{c} XB \hspace{8mm}
B \goes{e} C \hspace{8mm}
C \goes{c} CD \hspace{8mm}
C \goes{r} \epsilon \hspace{8mm}
D \goes{r} \epsilon
$$
then it is easy to see that $X'$ is regular if and only if
$L(X)= \emptyset$.
Obviously, $\Delta'$ is still visibly BPA. Hence regularity checking
on vBPA is P-hard for any equivalence between
trace equivalence and bisimilarity. 
\end{proof}

\begin{lem} \label{lem:NLcompl}
Regularity checking of v1CA w.r.t. any equivalence
between trace equivalence and bisimilarity
is NL-complete. 
\end{lem}
\begin{proof}
NL-hardness follows immediately from the fact
that the regularity checking problem naturally contains the reachability
problem on finite state systems (by a similar construction as 
in Lemma~\ref{lem:Phard}). 
The containment in nondeterministic logarithmic space
is by the observation that a given
visibly one-counter process $q_0(0)$ in $\Delta$, where $\Delta$
has $n$ control states, provides unbounded popping if
and only if there exist two control states $p$ and $p'$ such that 
\begin{enumerate}[(1)]
\item $q_0(0) \goes{w_1} p(n_1)$ for some $w_1$ and $n_1$ such that
$n_1 \geq n$ and $h(w') \leq n^2 + 2n$ for every prefix $w'$ of $w_1$, 
\item $p(n) \goes{w_2} p(n_2)$ for some $w_2$ and $n_2$ such that
$n_2 > n$ and $|w_2| \leq n$, 
\item $p(n) \goes{w_3} p'(n_3)$ for some $w_3$ and $n_3$
such that $|w_3| \leq n$, and
\item $p'(n) \goes{w_4} p'(n_4)$ for some $w_4$ and $n_4$ 
such that $n_4 < n$ and $|w_4| \leq n$. 
\end{enumerate}
Note that due to the restrictions on the lengths of the action
sequences in points 2., 3. and 4., no transition is performed
from any configuration where the counter value is zero.
Hence the same computations are possible also for any higher initial
counter value.

The idea is that the process $q_0(0)$ provides unbounded popping
iff  there is a possibility to arbitrarily increase the counter value 
(by means of the cycle from the control state $p$ in condition 2.) 
and then reach a control state $p'$ (condition 3.)
in which the counter value can be arbitrarily decreased 
(condition 4.). Initially, condition 1. guarantees that the state $p$
can be reached with a sufficiently large counter value.
The extra requirement $h(w') \leq n^2 +2n$ for any prefix $w'$ of $w_1$
in condition 1. is harmless because if the state $p(n_1)$ is reachable
then it is not necessary that the counter value during the computation
grows to more than $n^2 +2n$. To show that, we first observe that
we can require that the counter value $n_1$ satisfies $n \leq n_1 \leq 2n$.
For the sake of contradiction, let the control state $p$
be reachable from the initial configuration such that the minimal
counter value $n_1$ is greater than $2n$. We will show that we can then
reach $p$ with a counter value strictly smaller (while still at least $n$). 
Let us consider
the suffix of this computation where all configurations have
the counter values greater or equal to $n$. Now, in the region of configurations
with the counter values between $n$ and $2n$, there are necessarily
two configurations $r(n')$ and $r(n'')$ for some control state $r$
such that $n \leq n' < n''  \leq 2n$
and $r(n')$ precedes $r(n'')$. By removing the part of the computation
between $r(n')$ and $r(n'')$ we achieve a computation that reaches the
control state $p$ with a strictly smaller counter value.

We can hence assume that $p(n_1)$ is reachable such that 
$n \leq n_1 \leq 2n$. 
Should the counter grow to more than $n^2 +2n$ during this computation
then there would necessarily appear 
two configurations with the same counter value (greater than $2n$)
and the same control state and hence we could find a shorter sequence
of actions to reach $p(n_1)$. 

We shall now argue that the extra restrictions in conditions
2., 3. and 4. are harmless too. In condition 2., for the sake
of contradiction, assume that from the control state $p$
we can reach $p$ with a higher counter value (and never test
for zero during the computation) such that the shortest sequence
of actions to achieve this is strictly longer than $n$. On such
a sequence, there are necessarily two configurations $r(n')$ and
$r(n'')$ with the same control state $r$
such that $r(n')$ precedes $r(n'')$. If
$n' \geq n''$  than we can simply remove the part of the computation between 
these two configurations and reach the control state $p$ with a possibly
even higher counter value than before. If $n' < n''$ then we could
have initially selected the control state $r$ instead of $p$, because
there is a loop on the control state $r$ which increases the counter value.
Similarly, we can show that the restrictions in points 3. and 4. are
harmless too.
 
Finally, we note that
the control states $p$ and $p'$ above
can be nondeterministically guessed and the conditions 1.--4.
verified in nondeterministic logarithmic space. Hence the regularity
checking problem for visibly one-counter automata is in NL.
\end{proof}

We finish by a theorem summarizing the complexity results
proved in Lemma~\ref{lem:inP}, Lemma~\ref{lem:Phard} and
Lemma~\ref{lem:NLcompl}.

\begin{thm}
The regularity checking problem w.r.t. any equivalence
between trace equivalence and bisimilarity
(in particular also w.r.t. any equivalence considered in this article)
is P-complete for vPDA and vBPA and NL-complete for v1CA.
\end{thm}

\section{Conclusion} \label{sec:conclusion}
In the following table we provide a comparison of 
bisimulation, simulation and regularity (w.r.t. bisimilarity)
checking on PDA, 1CA, BPA and their subclasses vPDA, v1CA, vBPA. 
Results achieved in this article are in bold. \\

\begin{center}
\renewcommand{\arraystretch}{2.2}
\begin{tabular}{|c||c|c|c|}
\hline
& $\sim$ & $\simul$ and $\simuleq$ & $\sim$-regularity\\
\hline
\hline
PDA & \onthetop{\onthetop{decidable}{\cite{Senizergues:SIAM:05}}}
         {\onthetop{EXPTIME-hard}{\cite{KucMay:MFCS:02}}} & 
         \onthetop{undecidable}{\cite{GH:BPA-other}} & 
      \onthetop{?}
         {\onthetop{EXPTIME-hard}{\cite{KucMay:MFCS:02,Srba:ICALP:2002}}} \\
\hline
vPDA & 
\onthetop{\bf in EXPTIME}{\onthetop{EXPTIME-hard}{\cite{KucMay:MFCS:02}}} 
& \onthetop{\bf in EXPTIME}{\onthetop{EXPTIME-hard}{\cite{KucMay:MFCS:02}}} & 
{\bf P-complete} \\ 
\hline
1CA & \onthetop{decidable~\cite{Jancar:IC:00}}{\bf PSPACE-hard} &
\onthetop{undecidable}{\cite{JMS:99}}
 &  \onthetop{\onthetop{decidable}{\cite{Jancar:IC:00}}}
{\onthetop{P-hard}{\cite{BGS:92,Srba:ICALP:2002}}} \\
\hline
v1CA & {\bf PSPACE-complete} & {\bf PSPACE-complete} & 
{\bf NL-complete} 
\\
\hline
BPA &  \onthetop{\onthetop{in $2$-EXPTIME}{\cite{BCS:elementaryBPA}}}
         {\onthetop{PSPACE-hard}{\cite{Srba:ICALP:2002}}} & 
         \onthetop{undecidable}{\cite{GH:BPA-other}}
 & 
\onthetop{\onthetop{in $2$-EXPTIME}
         {\cite{BCS:BPA-regularity,BCS:elementaryBPA}}}
         {\onthetop{PSPACE-hard}{\cite{Srba:ICALP:2002}}} \\
\hline
vBPA & \onthetop{\bf in P}{\onthetop{P-hard}{\cite{BGS:92}}}
& \onthetop{\bf in P}{\onthetop{P-hard}{\cite{SawaJancar:PTIME-hard:CaI:2005}}} & 
{\bf P-complete}
 \\
\hline
\end{tabular}
\end{center}

\mbox{ } \\

In fact, our results about 
EXPTIME-completeness for vPDA, PSPACE-completeness
for v1CA and P-completeness for vBPA hold for all
preorders and equivalences between simulation preorder and
bisimulation equivalence studied in the literature
(like completed simulation, ready simulation and 2-nested
simulation).
The results confirm a general trend seen in the classical language
theory of pushdown automata: a relatively minor restriction
(from the practical point of view) of being able to distinguish 
call, return and internal actions often significantly improves
the complexity of the studied problems and sometimes even
changes undecidable problems into decidable ones, moreover with reasonable
complexity bounds. 

All the upper bounds proved in this article are matched by the corresponding
lower bounds. Here the most interesting result is PSPACE-hardness
of preorder/equivalence checking on v1CA for all relations between
simulation preorder and bisimulation equivalence. As noted in
Remark~\ref{rem:1C}, this proof improves also a number of other
complexity lower bounds for problems on standard one-counter nets and
one-counter automata, which were previously known to be only
DP-hard (DP-hardness is, likely, only a slightly stronger result than 
NP and co-NP hardness).

Finally, we have proved that for all the studied equivalences,
the regularity problem is decidable in polynomial time, more precisely,
P-complete for vPDA and vBPA and NL-complete for v1CA. Checking 
whether an infinite state process is equivalent to some 
regular one is a relevant question because many problems 
about such a process
can be answered by verifying the equivalent finite state system
and for finite state systems many efficient algorithms have been developed. 
A rather interesting observation is that preorder/equivalence 
checking on vBPA for preorders/equivalences between simulation 
and bisimilarity can be polynomially
translated to verification problems on finite state systems.
On the other hand, the class of vBPA processes is more expressive 
than the class of finite state processes and hence the question
whether for a given vPDA (or v1CA) process there is some equivalent vBPA
process is of a particular interest.
Another open problem is whether
the unbounded popping property for visibly pushdown automata can be
generalized so that it characterizes regularity also on vPDA
that can perform return actions even on the empty stack.

In the present article we did not consider any weak preorder/equivalences
as non-observable pushing and popping actions will immediately break
the synchronization of the stacks of the processes and the visibility 
constraint would not be usable any more.

\section*{Acknowledgement}
I would like to thank Markus Lohrey
for a discussion at ETAPS'06 and for a reference to 
PSPACE-completeness of the emptiness problem
for alternating automata over a one-letter alphabet. 
My thanks go also to the anonymous referees of CSL'06 and LMCS 
for their useful comments and for suggesting the P-hardness proof 
of regularity checking for vBPA.


\begin{thebibliography}{AKMV05}

\bibitem[AEM04]{AEM:04}
R.~Alur, K.~Etessami, and P.~Madhusudan.
\newblock A temporal logic of nested calls and returns.
\newblock In {\em Proceedings of the 10th International Conference on Tools and
  Algorithms for the Construction and Analysis of Systems (TACAS'04)}, volume
  2988 of {\em LNCS}, pages 467--481. Springer-Verlag, 2004.

\bibitem[AFI01]{AFI:01}
L.~Aceto, W.~Fokkink, and A.~Ing\'{o}lfsd\'{o}ttir.
\newblock 2-nested simulation is not finitely equationally axiomatizable.
\newblock In {\em Proceedings of the 18th Annual Symposium on Theoretical
  Aspects of Computer Science (STACS'01)}, volume 2010 of {\em LNCS}, pages
  39--50. Springer-Verlag, 2001.

\bibitem[AKMV05]{AKMV:06}
R.~Alur, V.~Kumar, P.~Madhusudan, and M.~Viswanathan.
\newblock Congruences for visibly pushdown languages.
\newblock In {\em Proceedings of the 32nd International Colloquium on Automata,
  Languages and Programming (ICALP'05)}, volume 3580 of {\em LNCS}, pages
  1102--1114. Springer-Verlag, 2005.

\bibitem[AM04]{AM:vPDA:04}
R.~Alur and P.~Madhusudan.
\newblock Visibly pushdown languages.
\newblock In {\em Proceedings of the 36th Annual ACM Symposium on Theory of
  Computing (STOC'04)}, pages 202--211. ACM Press, 2004.

\bibitem[BCS95]{BCS:elementaryBPA}
O.~Burkart, D.~Caucal, and B.~Steffen.
\newblock An elementary decision procedure for arbitrary context-free
  processes.
\newblock In {\em Proceedings of the 20th International Symposium on
  Mathematical Foundations of Computer Science (MFCS'95)}, volume 969 of {\em
  LNCS}, pages 423--433. Springer-Verlag, 1995.

\bibitem[BCS96]{BCS:BPA-regularity}
O.~Burkart, D.~Caucal, and B.~Steffen.
\newblock Bisimulation collapse and the process taxonomy.
\newblock In {\em Proceedings of the 7th International Conference on
  Concurrency Theory (CONCUR'96)}, volume 1119 of {\em LNCS}, pages 247--262.
  Springer-Verlag, 1996.

\bibitem[BEM97]{CONCUR::BouajjaniEM1997}
A.~Bouajjani, J.~Esparza, and O.~Maler.
\newblock Reachability analysis of pushdown automata: Application to
  model-checking.
\newblock In {\em Proceedings of the 8th International Conference on
  Concurrency Theory (CONCUR'97)}, volume 1243 of {\em LNCS}, pages 135--150.
  Springer-Verlag, 1997.

\bibitem[BGS92]{BGS:92}
J.~Balcazar, J.~Gabarro, and M.~Santha.
\newblock Deciding bisimilarity is {P}-complete.
\newblock {\em Formal Aspects of Computing}, 4(6A):638--648, 1992.

\bibitem[BLS06]{BLS:06}
V.~B\'ar\'any, Ch. L\"oding, and O.~Serre.
\newblock Regularity problems for visibly pushdown languages.
\newblock In {\em Proceedings of the 23rd Annual Symposioum on Theoretical
  Aspects of Computer Science (STACS'06)}, volume 3884 of {\em LNCS}, pages
  420--431. Springer-Verlag, 2006.

\bibitem[Cau06]{Caucal:DLT:06}
D.~Caucal.
\newblock Synchronization of pushdown automata.
\newblock In {\em Proceedings of 10th International Conference on Developments
  in Laguage Theory (DLT'06)}, volume 4036 of {\em LNCS}, pages 120--132.
  Springer-Verlag, 2006.

\bibitem[EHRS00]{EHRS00}
J.~Esparza, D.~Hansel, P.~Rossmanith, and S.~Schwoon.
\newblock Efficient algorithms for model checking pushdown systems.
\newblock In {\em Proceedings of the 12th International Conference on Computer
  Aided Verification (CAV'00)}, volume 1855 of {\em LNCS}, pages 232--247.
  Springer-Verlag, 2000.

\bibitem[FP01]{FisPnu:01}
D.~Fisman and A.~Pnueli.
\newblock Beyond regular model checking.
\newblock In {\em Foundations of Software Technology and Theoretical Computer
  Science ({FST\&TCS}'01)}, volume 2245 of {\em LNCS}, pages 156--170.
  Springer-Verlag, 2001.

\bibitem[GH94]{GH:BPA-other}
J.F. Groote and H.~H{\"{u}}ttel.
\newblock Undecidable equivalences for basic process algebra.
\newblock {\em Information and Computation}, 115(2):353--371, 1994.

\bibitem[GHR95]{GHR:Pcompleteness}
R.~Greenlaw, H.~J. Hoover, and W.~L. Ruzzo.
\newblock {\em Limits to Parallel Computation: ${P}$-Completeness Theory}.
\newblock Oxford University Press, 1995.

\bibitem[Hol96]{Holzer:95}
M.~Holzer.
\newblock On emptiness and counting for alternating finite automata.
\newblock In {\em Proceedings of the 2nd International Conference on
  Developments in Language Theory (DLT'95)}, pages 88--97. World Scientific,
  1996.

\bibitem[Jan00]{Jancar:IC:00}
P.~Jan\v{c}ar.
\newblock Decidability of bisimilarity for one-counter processes.
\newblock {\em Information and Computation}, 158(1):1--17, 2000.

\bibitem[JKMS04]{JKMS:04}
P.~Jan\v{c}ar, A.~Ku\v{c}era, F.~Moller, and Z.~Sawa.
\newblock {DP} lower bounds for equivalence-checking and model-checking of
  one-counter automata.
\newblock {\em Information and Computation}, 188(1):1--19, 2004.

\bibitem[JMS99]{JMS:99}
P.~Jan\v{c}ar, F.~Moller, and Z.~Sawa.
\newblock Simulation problems for one-counter machines.
\newblock In {\em Proceedings of the 26th Annual Conference on Current Trends
  in Theory and Practice of Informatics (SOFSEM'99)}, volume 1725 of {\em
  LNCS}, pages 404--413. Springer-Verlag, 1999.

\bibitem[JS07]{JS:1L-AFA}
P.~Jan\v{c}ar and Z.~Sawa.
\newblock A note on emptiness for alternating finite automata with a one-letter
  alphabet.
\newblock {\em Information Processing Letters}, 104(5):164--167, 2007.

\bibitem[JS08]{JS:08:JACM}
P.~Jan\v{c}ar and J.~Srba.
\newblock Undecidability of bisimilarity by defender's forcing.
\newblock {\em Journal of the ACM}, 55(1):5:1--5:26, 2008.

\bibitem[KM02a]{KucMay:MFCS:02}
A.~Ku\v{c}era and R.~Mayr.
\newblock On the complexity of semantic equivalences for pushdown automata and
  {BPA}.
\newblock In {\em Proceedings of the 27th International Symposium on
  Mathematical Foundations of Computer Science (MFCS'02)}, volume 2420 of {\em
  LNCS}, pages 433--445. Springer-Verlag, 2002.

\bibitem[KM02b]{KM:RS-2002-01}
A.~Ku\v{c}era and R.~Mayr.
\newblock On the complexity of semantic equivalences for pushdown automata and
  {BPA}.
\newblock Technical Report FIMU-RS-2002-01, Faculty of Informatics, Masaryk
  University, 2002.

\bibitem[Koz83]{Kozen:mu-calculus}
D.~Kozen.
\newblock Results on the propositional $\mu$-calculus.
\newblock {\em Theoretical Computer Science}, 27:333--354, 1983.

\bibitem[LMS04]{LMS:FSTTCS:04}
Ch. L{\"o}ding, P.~Madhusudan, and O.~Serre.
\newblock Visibly pushdown games.
\newblock In {\em Proceedings of the 24th International Conference on
  Foundations of Software Technology and Theoretical Computer Science
  (FSTTCS'04)}, volume 3328 of {\em LNCS}, pages 408--420. Springer-Verlag,
  2004.

\bibitem[MW05]{MW:05}
A.~Murawski and I.~Walukiewicz.
\newblock Third-order idealized algol with iteration is decidable.
\newblock In {\em Proceedings of the 8th International Conference on
  Foundations of Software Science and Computation Structures (FOSSACS'05)},
  volume 3441 of {\em LNCS}, pages 202--218, 2005.

\bibitem[NS07]{NS:MFCS:07}
D.~Nowotka and J.~Srba.
\newblock Height-deterministic pushdown automata.
\newblock In {\em Proceedings of 32nd International Symposium on Mathematical
  Foundations of Computer Science (MFCS'07)}, volume 4708 of {\em LNCS}, pages
  125--134. Springer-Verlag, 2007.

\bibitem[Pit05]{Pitcher:05}
C.~Pitcher.
\newblock Visibly pushdown expression effects for {XML} stream processing.
\newblock In {\em Proceedings of Programming Language Technologies for {XML}
  (PLAN-X)}, pages 5--19, 2005.

\bibitem[PT87]{PaigeTarjan87}
R.~Paige and R.~E. Tarjan.
\newblock Three partition refinement algorithms.
\newblock {\em SIAM Journal of Computing}, 16(6):973--989, December 1987.

\bibitem[S{\'e}n98]{PDA:bisimilarity}
G.~S{\'e}nizergues.
\newblock Decidability of bisimulation equivalence for equational graphs of
  finite out-degree.
\newblock In {\em Proceedings of the 39th Annual Symposium on Foundations of
  Computer Science({FOCS}'98)}, pages 120--129. IEEE Computer Society, 1998.

\bibitem[S{\'e}n01]{Senizergues:TCS2001}
G.~S{\'e}nizergues.
\newblock {L(A)=L(B)?} {D}ecidability results from complete formal systems.
\newblock {\em Theoretical Computer Science}, 251(1--2):1--166, 2001.

\bibitem[S{\'e}n02]{Senizergues:TCS2001_simple}
G.~S{\'e}nizergues.
\newblock {L(A)=L(B)?} {A} simplified decidability proof.
\newblock {\em Theoretical Computer Science}, 281(1--2):555--608, 2002.

\bibitem[S{\'e}n05]{Senizergues:SIAM:05}
G.~S{\'e}nizergues.
\newblock The bisimulation problem for equational graphs of finite out-degree.
\newblock {\em SIAM Journal on Computing}, 34(5):1025--1106, 2005.

\bibitem[Ser06]{Serre:06}
O.~Serre.
\newblock Parity games played on transition graphs of one-counter processes.
\newblock In {\em Proceedings of the 9th International Conference on
  Foundations of Software Science and Computation Structures (FOSSACS'06)},
  volume 3921 of {\em LNCS}, pages 337--351. Springer-Verlag, 2006.

\bibitem[SJ05]{SawaJancar:PTIME-hard:CaI:2005}
Z.~Sawa and P.~Jan{\v{c}}ar.
\newblock Behavioural equivalences on finite-state systems are {PTIME}-hard.
\newblock {\em Computing and Informatics}, 24:513 -- 528, 2005.

\bibitem[Srb01]{Srba:concur:2001}
J.~Srba.
\newblock On the power of labels in transition systems.
\newblock In {\em Proceedings of the 12th International Conference on
  Concurrency Theory (CONCUR'01)}, volume 2154 of {\em LNCS}, pages 277--291.
  Springer-Verlag, 2001.

\bibitem[Srb02]{Srba:ICALP:2002}
J.~Srba.
\newblock Strong bisimilarity and regularity of basic process algebra is
  {PSPACE}-hard.
\newblock In {\em Proceedings of the 29th International Colloquium on Automata,
  Languages and Programming (ICALP'02)}, volume 2380 of {\em LNCS}, pages
  716--727. Springer-Verlag, 2002.

\bibitem[Ste67]{Stearns:67:IC}
R.~E. Stearns.
\newblock A regularity test for pushdown machines.
\newblock {\em Information and Control}, 11(3):323--340, 1967.

\bibitem[Sti95]{Stirling:CONCUR95}
C.~Stirling.
\newblock Local model checking games.
\newblock In {\em Proceedings of the 6th International Conference on
  Concurrency Theory (CONCUR'95)}, volume 962 of {\em {LNCS}}, pages 1--11.
  Springer-Verlag, 1995.

\bibitem[Sti01]{Stirling:TCS2001}
C.~Stirling.
\newblock Decidability of {DPDA} equivalence.
\newblock {\em Theoretical Computer Science}, 255(1--2):1--31, 2001.

\bibitem[Tho93]{Thomas1993TAPSOFT}
W.~Thomas.
\newblock On the {E}hrenfeucht-{F}ra{\"\i}ss{\'e} game in theoretical computer
  science (extended abstract).
\newblock In {\em Proceedings of the 4th International Joint Conference
  {CAAP}/{FASE}, Theory and Practice of Software Development (TAPSOFT'93)},
  volume 668 of {\em {LNCS}}, pages 559--568. Springer-Verlag, 1993.

\bibitem[Val75]{Valiant:75:JACM}
L.G. Valiant.
\newblock Regularity and related problems for deterministic pushdown automata.
\newblock {\em Journal of the ACM}, 22(1):1--10, 1975.

\bibitem[vG90a]{Glabbeek:PhD}
R.J. van Glabbeek.
\newblock {\em Comparative Concurrency Semantics and Refinement of Actions}.
\newblock PhD thesis, CWI/Vrije Universiteit, 1990.

\bibitem[vG90b]{Glabbeek:hierarchy}
R.J. van Glabbeek.
\newblock The linear time---branching time spectrum.
\newblock In {\em Proceedings of the 1st Internatinal Conference on Theories of
  Concurrency: Unification and Extension (CONCUR'90)}, volume 458 of {\em
  LNCS}, pages 278--297. Springer-Verlag, 1990.

\bibitem[Wal01]{Wal:PDA:01}
I.~Walukiewicz.
\newblock Pushdown processes: Games and model-checking.
\newblock {\em Information and Computation}, 164(2):234--263, 2001.

\end{thebibliography}

\end{document}